\documentclass[12pt,a4paper]{article}
\usepackage[latin1,utf8]{inputenc}
\usepackage{amsmath}
\usepackage{amsthm}
\usepackage{amsfonts}
\usepackage{amssymb}
\usepackage{thmtools}
\usepackage{thm-restate}
\usepackage{hyperref}
\usepackage{cleveref}
\usepackage{graphicx}
\usepackage{nicefrac}
\usepackage{fullpage}
\usepackage{bm}
\usepackage{xcolor}
\usepackage[linesnumbered,ruled]{algorithm2e}
\usepackage{tikz}
\usepackage{enumitem}

\newtheorem{thm}{Theorem}[section]
\newtheorem{prop}[thm]{Proposition}
\newtheorem{lemma}[thm]{Lemma}
\newtheorem{remark}[thm]{Remark}

\newtheorem{claim}[thm]{Claim}
\newtheorem{defi}[thm]{Definition}
\newtheorem{cnst}[thm]{Construction}

\newcommand{\Fq}{\mathbb{F}_q}
\newcommand{\Fp}{\mathbb{F}_p}

\newcommand{\cC}{\mathcal{C}}

\newcommand{\bX}{{\bf X}}

\newcommand\blfootnote[1]{%
  \begingroup
  \renewcommand\thefootnote{}\footnote{#1}%
  \addtocounter{footnote}{-1}%
  \endgroup
}

\DeclareMathOperator{\poly}{poly}
\newcommand{\F}{\mathbb{F}}
\newcommand{\cR}{\mathcal{R}}
\newcommand{\ed}{\text{ED}}
\DeclareMathOperator{\charac}{char}
\begin{document}
	
	\title{Reed Solomon Codes Against Adversarial Insertions and Deletions}
\author{
	Roni Con\thanks{Blavatnik School of Computer Science, Tel Aviv University, Tel Aviv, Israel. Email: roni.con93@gmail.com} \and 
	Amir Shpilka \thanks{Blavatnik School of Computer Science, Tel Aviv University, Tel Aviv, Israel. Email: shpilka@tauex.tau.ac.il. The research leading to these results has received funding from the Israel Science Foundation (grant number 514/20) and from the Len Blavatnik and the Blavatnik Family Foundation.} \and
	Itzhak Tamo \thanks{Department of EE-Systems, Tel Aviv University, Tel Aviv, Israel. Email: zactamo@gmail.com.}
}
\date{}
\maketitle
\begin{abstract}
\blfootnote{The work of Itzhak Tamo and Roni Con was partially supported by the European Research Council (ERC grant number 852953) and by the Israel Science Foundation (ISF grant number 1030/15).}
	\sloppy 
	In this work, we  study the performance of Reed--Solomon codes against adversarial  insertion-deletion (insdel) errors.

	We prove that over fields of size $n^{O(k)}$ there are $[n,k]$ Reed-Solomon codes that can decode from $n-2k+1$ insdel errors and hence attain  the half-Singleton bound. We also give a deterministic construction of such codes over much larger fields (of size  $n^{k^{O(k)}}$).  Nevertheless, for $k=O(\log n /\log\log n)$ our construction runs in polynomial time. For the special case $k=2$, which received a lot of attention in the literature, we construct an $[n,2]$ Reed-Solomon code over a field of size $O(n^4)$ that can decode from $n-3$ insdel errors. Earlier  constructions required an exponential field size.
	Lastly, we  prove that any such construction requires a field of size $\Omega(n^3)$.

\end{abstract}
\thispagestyle{empty}
\newpage

\section{Introduction}
    Error-correcting codes are among the most widely used tools and objects of study in information theory and theoretical computer science. The most common model of corruption that is studied in the TCS literature is  that of errors or erasures. The model in which each symbol of the transmitted word is either replaced with a different symbol from the alphabet (an error) or with a `?' (an erasure). The theory of such codes began with the seminal work of Shannon, \cite{shannon1948mathematical}, who studied random errors and erasures and the work of Hamming \cite{hamming1950error} who studied the adversarial model for errors and erasures. These models are mostly well understood, and today we know efficiently encodable and decodable codes that are optimal for Shannon's model of random errors. For adversarial errors, we have optimal codes over large alphabets and  good codes (codes of constant relative rate and relative distance) for every constant sized alphabet.  
    
    Another important model that has been considered ever since Shannon's work is that of \emph{synchronization} errors. These are errors that affect the length of the received word. The most common model for studying synchronization errors is the insertion-deletion model  (insdel for short): an insertion error is when a new symbol is inserted between two symbols of the transmitted word. A deletion is when a symbol is removed from the transmitted word. For example, over the binary alphabet, when $100110$ is transmitted, we may receive the word $1101100$, which is obtained from two insertions ($1$ at the beginning and $0$ at the end) and one deletion (one of the $0$'s at the beginning of the transmitted word).  
    Observe that compared to the more common error model, if an adversary wishes to \emph{change} a symbol, then the cost is that of two operations - first deleting the symbol and then inserting a new one instead.
	
	Insdel errors appear in diverse settings such as optical recording, semiconductor devices, integrated circuits, and synchronous digital communication networks. Another important example is the trace reconstruction problem, which has applications in computational biology and DNA-based storage systems \cite{bornholt2016dna,yazdi2017portable,heckel2019characterization}. See the surveys \cite{mitzenmacher2009survey,mercier2010survey} for a good picture of the problems and applications of error-correcting codes for the insdel model (insdel codes for short). 

Reed-Solomon codes are the most widely used family of codes in theory and practice. Indeed, they have found many applications both in theory and in practice (their applications include QR codes \cite{soon2008qr}, secret sharing schemes \cite{mceliece1981sharing}, space transmission \cite{wicker1999reed}, encoding data on CDs  \cite{wicker1999reed} and more. The ubiquity of these codes can be attributed to their  simplicity as well as to their efficient encoding and decoding algorithms. 
	As such, it is an important problem to understand whether they can also decode from insdel errors. This problem  received a lot of attention recently \cite{safavi2002traitor,wang2004deletion,tonien2007construction,duc2019explicit,liu20212,chen2021improved,liu2021bounds}, but besides very few constructions (i.e., evaluation points for Reed-Solomon codes), not much was known before our work. We discuss this line of work in more detail in Section~\ref{sec:prev-results}.

In this paper, we first prove that there are Reed-Solomon codes that achieve the half-Singleton bound. In other words, there are optimal Reed-Solomon codes also against insdel errors. 
	We also give a set of evaluation points that define a Reed-Solomon code that achieves this bound. As the field size that we get grows very fast, our construction runs in polynomial time only for very small values of $\delta$.
	We also explicitly construct  $2$-dimensional RS codes over a  field size smaller than the previous known constructions.
	Unfortunately, we could not develop efficient decoding algorithms for our Reed-Solomon constructions, and we leave this as an open problem for future research.
	
\subsection{Basic definitions and notation}
	
	For an integer $k$, we denote $[k]=\{1,2,\ldots,k\}$. 
	Throughout this paper, $\log(x)$ refers to the base-$2$ logarithm. For a prime power $q$, we denote with $\F_q$ the field of size $q$.

	We denote the $i$th symbol of a string $s$ (or of a vector $v$) as  $s_i$ (equivalently $v_i$). Throughout this paper, we shall move freely between representations of vectors as strings and vice versa. Namely, we shall view each vector $v=(v_1, \ldots, v_n)\in \Fq^n$ also as a string by concatenating all the symbols of the vector into one string, i.e., $(v_1, \ldots, v_n) \leftrightarrow v_1 \circ v_2 \circ \ldots \circ v_n$. Thus, if we say that $s$ is a subsequence of some vector $v$, we mean that we view $v$ as a string and $s$ is a subsequence of that string. 
	
	An error correcting code of block length $n$  over an alphabet $\Sigma$ is a subset $\cC\subseteq \Sigma^n$. The rate of $\cC$ is $\frac{\log|\cC|}{n\log|\Sigma|}$, which  captures the amount of information encoded in every symbol of a codeword.	
	A linear code over a field $\F$ is a linear subspace $\cC\subseteq \F^n$. The rate of a linear code  $\cC$ of block length $n$ is $\cR=\dim(\cC)/n$.  Every linear code of dimension $k$ can be  described as the image of a linear  map, which, abusing notation, we also denote with $\cC$, i.e., $\cC : \F^k \rightarrow \F^n$. Equivalently, a linear code $\cC$ can be defined by a \emph{parity check matrix} $H$ such that $x\in\cC$ if and only if $Hx=0$. When $\cC\subseteq \F_q^n$ has dimension $k$  we say that it is an $[n,k]_q$ code. The minimal distance of $\cC$ with respect to a metric $d(\cdot,\cdot)$ is defined as $\text{dist}_{\cC}:= \min_{v\neq u \in \cC}{d(v,u)}$. 
	Naturally, we would like the rate to be as large as possible, but there is an inherent tension between the rate of the code and the minimal distance (or the number of errors that a code can decode from).
	In this work, we focus on codes against insertions and deletions. 
	\begin{defi}
		Let $s$ be a string over the alphabet $\Sigma$. The operation in which we remove a symbol from $s$ is called a \emph{deletion} and the operation in which we place a new symbol from $\Sigma$ between two consecutive symbols in $s$, in the beginning, or at the end of $s$, is called an \emph{insertion}. 
		
		A \emph{substring} of $s$ is a string obtained by taking consecutive symbols from $s$.
		A \emph{subsequence} of $s$ is a string obtained by removing some (possibly none) of the symbols in $s$. 
	\end{defi}

	The relevant metric for such codes is the edit-distance that we define next.
	
	\begin{defi}
	
		Let $s,s'$ be strings over the alphabet $\Sigma$. 
		A \emph{longest common subsequence} between $s$ and $s'$, is a subsequence $s_\textup{sub}$ of both $s$ and $s'$, of maximal length. We denote by $ \textup{LCS}(s,s')$ the length of a longest common subsequence.\footnote{Note that a longest common subsequence may not be unique as there can be a number of subsequences of maximal length. For example in the strings $s=(1,0)$ and $s'=(0,1)$.} 
		
		The \emph{edit distance} between $s$ and $s'$, denoted by $\ed(s,s')$, is the minimal number of insertions and deletions needed in order to turn $s$ into $s'$. One can verify that this measure  indeed defines a metric (distance function). 
	\end{defi}

	\begin{lemma}[See e.g.  Lemma 12.1 in \cite{crochemore2003jewels}]\label{lem:lcs}
		It holds that $\textup{ED}(s,s') = \left|s\right| + \left|s' \right| - 2 \textup{LCS}(s,s') $.
	\end{lemma}

	We next define Reed-Solomon codes (RS-codes from now on). 
	
	\begin{defi}[Reed-Solomon codes]
		Let $\alpha_1, \alpha_2, \ldots, \alpha_n \in\F_q$ be distinct points in a finite field $\mathbb{F}_q$ of order  $q\geq n$. For $k\leq n$ the $[n,k]_q$ RS-code 
		defined  by the  evaluation set $\{ \alpha_1, \ldots, \alpha_n \}$ is the   set of codewords 
		\[
		\left \lbrace c_f = \left( f(\alpha_1), \ldots, f(\alpha_n) \right) \mid f\in \mathbb{F}_q[x],\deg f < k \right \rbrace \;.
		\]
	\end{defi}
	
	In words, a codeword of an $[n,k]_q$ RS-code is the evaluation vector of some polynomial of degree less than $k$ at $n$ predetermined distinct points. 
	It is well known (and easy to see) that the rate of $[n,k]_q$ RS-code is $k/n$ and the minimal distance, with respect to the Hamming metric,  is $n-k+1$.

\subsection{Previous results}
	\label{sec:prev-results}
	
	Linear codes against worst-case insdel errors were recently studied by
	Cheng, Guruswami, Haeupler, and Li \cite{cheng2020efficient}. Correcting an error in a preceding work, they proved that there are good linear codes against insdel errors. 
	\begin{thm}[Theorem 4.2 in \cite{cheng2020efficient}] \label{thm:random-code}
		For any $\delta > 0$ and prime power $q$, there exists a family of linear codes over $\Fq$ that can correct up to $\delta n$ insertions and deletions, with rate $(1-\delta)/2 - h(\delta)/\log_2 (q)$.
	\end{thm} 
	The proof of \Cref{thm:random-code} uses the probabilistic method, showing that, with high probability, a random linear map generates such code. Complementing their result, they proved that their construction is  almost tight. Specifically, they provided the following upper bound, which they call ``half-Singleton bound,'' that holds over any field.
	\begin{thm}[Half-Singleton bound: Corollary 5.1 in \cite{cheng2020efficient}]
		Every linear insdel code which is capable of correcting a $\delta$ fraction of deletions has rate at most $(1-\delta)/2 + o(1)$.
	\end{thm}

	The performance of RS-codes against insdel errors was studied much earlier than the recent work of Cheng et al. \cite{cheng2020efficient}.
	To the best of our knowledge,  Safavi-Naini and Wang \cite{safavi2002traitor} were the first to study the performance of RS-codes against insdel errors. They gave an algebraic condition that is sufficient for an RS-code to correct from insdel errors, yet they did not provide any construction. In fact, in our work, we  consider an almost identical algebraic condition, and by simply using the Schwartz-Zippel-Demillo-Lipton lemma, we prove that there are RS-codes that meet this condition and, in addition,   achieve the half-Singleton bound. In particular, RS-codes are optimal for insdel errors (see discussion in \Cref{sec:rs-insdel}). 
	Wang,  McAven, and  Safavi-Naini  \cite{wang2004deletion} constructed a $[5,2]$ RS-code capable of correcting a single deletion. Then, in \cite{tonien2007construction}, Tonien and  Safavi-Naini constructed an $[n,k]$ generalized-RS-codes capable of correcting from $\log_{k+1} n - 1$ insdel errors. Similar to our results, they did not provide an efficient decoding algorithm.

	In another line of  work
 	Duc, Liu, Tjuawinata, and Xing \cite{duc2019explicit},  Liu and Tjuawinata \cite{liu20212}, Chen and Zhang \cite{chen2021improved}, and Liu and Xing \cite{liu2021bounds}  studied the specific case of $2$-dimensional RS-codes. 
	    
	In \cite{duc2019explicit,liu20212}, the authors presented constructions of $[n,2]$ RS-codes that for every $\varepsilon>0$ can correct from $(1-\varepsilon)\cdot n$ insdel errors, for codes of length $n=\poly(1/\varepsilon)$ over fields of size $\Omega(\exp((\log n)^{1/\varepsilon}))$ and $\Omega(\exp(n^{1/\varepsilon}))$, respectively. In \cite{duc2019explicit,chen2021improved}, the authors present constructions of two-dimensional RS-codes that can correct from $n-3$ insdel errors where the field size is exponential in $n$.  After a draft of this work appeared online, 
	Liu and Xing \cite{liu2021bounds} constructed, using different approach than us, a two dimensional RS-codes over that can correct from $n-3$ insdel errors, over a field size $O(n^5)$. Specifically, they prove the following.
	\begin{thm}\cite[Theorem 4.8]{liu2021bounds}
    Let $n\geq 4$. If $q > \frac{n(n-1)^2(n-2)^2}{4}$, then there is an $[n,2]_q$ RS-code, constructed in polynomial time, that can decode from $n-3$ insdel errors.
	\end{thm} 

\subsection{Our results}

    First, we prove that there are RS-codes that achieve the half-Singleton bound. Namely, they are optimal linear codes for insdel errors. 
	
	\begin{restatable}{thm}{rssz} \label{thm:rs-sz}
		Let $k$ and $n$ be positive integers such that $2k - 1 \leq n$. 
		For $q = O(n^{4k-2})$ there exists 
		an $[n,k]_q$ RS-code defined by $n$ distinct evaluation points $\alpha_1, \ldots, \alpha_n\in\Fq$, that can recover from $n - 2k + 1$ adversarial insdel errors.
		\end{restatable}
		

	
	Observe that the constructed code achieves the half Singleton bound: its rate is $\mathcal{R} = k/n = (1-\delta)/2 + o(1)$ and $\delta=(n - 2k + 1)/n$. 
	
	\Cref{thm:rs-sz} is an existential result and does not give an explicit construction. Using ideas from number theory and algebra, we construct RS-codes that can decode from $n-2k+1$ 
	adversarial insdel errors, in particular, they achieve the half-Singleton bound. Specifically,
	
	\begin{restatable}{thm}{rsExpConst} \label{thm:rs-explicit-const}
		Let $k$ and $n$ be positive integers, where $2k - 1 \leq n$. There is a deterministic construction of an $[n,k]_q$ RS-code that can correct from $n-2k+1$  insdel errors where $q = O\left(n^{k^2 \cdot ((2k)!)^2}\right)$. The construction runs in polynomial time for $k = O(\log(n)/\log(\log(n)))$.
	\end{restatable}
	
	We note that for $k = \omega(\log(n)/\log\log(n))$ the field size is $\exp(n^{\omega(1)})$ and in particular, there is no efficient way to represent arbitrary elements of $\F_q$ in this case.

As discussed before, special attention was given in the literature to the case of two dimensional RS-codes. By using Sidon spaces that were constructed in \cite{roth2017construction}, we  explicitly construct a family of $[n,2]_q$ RS-codes that can decode from $n-3$ insdel errors for $q=O(n^4)$. 
	Besides improving on all previous constructions in terms of field size, our construction also requires a smaller field size than the one  guaranteed by the randomized argument in \Cref{thm:rs-sz}. Such phenomena, where a deterministic algebraic   construction outperforms the parameters obtained by a randomized construction, are scarce  in coding theory and  combinatorics.
	Well-known examples are AG codes that outperform the GV-bound \cite{tsfasman1982modular} and   constructions of extremal graphs with ``many'' edges that do not contain cycles of length $4$, $6$ or $10$ (see \cite{conlon}).

	\begin{restatable}{thm}{rsTwoDimConst}
		\label{thm:rs-twodim-const}
		For any $n\geq 4$, there exists an explicit $[n,2]_{q}$ RS-code that can correct from $n - 3$  insdel errors, where $q = O(n^4)$. 
	\end{restatable}
	
	We also prove a (very) weak lower bound on the field size. 
	
	\begin{restatable}{prop}{RSLowerBound}
	\label{prop:lower-bound}
		Any $[n,k]_q$ RS-code that can correct from $n - 2k +1$ worst case insdel errors must satisfy
		\[q \geq \frac{1}{2} \cdot  \left(\frac{n}{(2k-1)(k-1)}\right)^{\frac{2k-1}{k-1}} \;. \]		\end{restatable}
	
	While for large values of $k$, this bound is meaningless, it implies that when $k=2$, the field size must be $\Omega(n^3)$. Thus, the construction given in \Cref{thm:rs-twodim-const} is nearly optimal.
	 The gap between the field size in our construction and the one implied by the lower bound raises an interesting question: what is the minimal field size $q$ for which an optimal $[n,2]_q$ RS-code exists?

    \subsection{Proof idea}
    	To show that RS-codes can be used against insdel errors, we first prove an  algebraic condition that is sufficient for $n$ evaluation points to define an RS-code that can decode from insdel errors. This condition requires that a certain set of $n^{O(k)}$ matrices, determined by the evaluation points, must all have full rank. Then, a simple application of the Schwartz-Zippel-DeMilo-Lipton lemma \cite{DBLP:journals/jacm/Schwartz80,Zippel79,DBLP:journals/ipl/DemilloL78} implies the existence of good evaluation points over fields of size $n^{O(k)}$. To obtain a deterministic construction, we show that by going to much larger field size, one can find evaluation points satisfying the full-rank condition. While the field size needs to be of size roughly $\Omega(n^{k^k})$, we note that, for not too large values of $k$, it is of exponential size, and in this case, our construction runs in polynomial time. A key ingredient in the analysis of this construction is our use of the `abc theorem' for polynomials over finite fields \cite{vaserstein2003vanishing}. 
    	
	For the case of $k=2$, we use a different idea that gives a better field size than the one implied by the probabilistic argument above. We do so by noting that in this case the full-rank condition can be expressed as the requirement that no two different  triples of  evaluation points $(x_1,x_2,x_3)$ and $(y_1,y_2,y_3)$ satisfy 
		\[
		\frac{y_1 - y_2}{x_1 - x_2} = \frac{y_2 - y_3}{x_2 - x_3} \;.
		\]
	This condition is reminiscent of the condition behind the construction of Sidon spaces of \cite{roth2017construction}, and indeed, we build on their construction of Sidon spaces to define good evaluation points in a field of size $O(n^4)$.

	\subsection{Organization}
	The paper is organized as follows.
	In \Cref{sec:rs-insdel}, we prove \Cref{thm:rs-sz}. 
	In \Cref{sec:det-any-k}, we prove \Cref{thm:rs-explicit-const}. Finally, in \Cref{sec:det-k-2}, we prove  \Cref{thm:rs-twodim-const} and \Cref{prop:lower-bound}. \Cref{sec:open-que} is devoted to conclusion and open questions.

\section{Reed-Solomon codes achieving the half-Singleton bound}
	\label{sec:rs-insdel}
	In this section, we  prove our results concerning RS-codes. Specifically, we prove that RS-codes achieve the half-Singleton bound and give some explicit constructions. The proofs  will follow by standard analysis of the LCS between any two distinct codewords. 
	
	We begin by reformulating the condition on the maximum length of an LCS as an algebraic condition (invertibility of certain matrices). Then we show that an RS-code that satisfies this condition  would have the maximum possible  edit distance and hence would be able to decode from the maximum number of insdel errors. 
	We remark that a similar approach already appeared in \cite[Section 2.2]{safavi2002traitor} and we shall repeat some of the details here.

\subsection{An algebraic condition}	

	%

	The following proposition is the main result of this section as it  provides a sufficient condition for an RS-code to recover from the maximum number of insdel errors. 
	We first make the following definitions:
		We say that a vector of indices $I\in [n]^s$ is an \emph{increasing} vector if its coordinates are monotonically increasing, i.e., for any  $1\leq i<j\leq s$, $I_i<I_j$, where $I_i$ is the $i$th coordinate of $I$.  For a codeword $c$ of length $n$ and an increasing vector $I$, let $c_I$ be the restriction of $c$ to the coordinates with indices in $I$, i.e., $c_I=(c_{I_1},\ldots,c_{I_s})$. 
	 For two vectors $I,J\in [n]^{2k-1}$  with distinct coordinates we define the following (variant of a) vandermonde matrix of order $(2k-1)\times (2k-1)$ in the formal variables $\bX=(X_1,\ldots,X_n)$:

    \begin{equation} \label{eq:mat-lcs-eq}
	V_{I,J}(\bX)=\begin{pmatrix} 
	1 & X_{I_1} & \ldots & X_{I_1}^{k-1}  & X_{J_1} &\ldots & X_{J_1}^{k-1} \\ 
	1 & X_{I_2} & \ldots & X_{I_2}^{k-1}  & X_{J_2} &\ldots & X_{J_2}^{k-1} \\
	\vdots &\vdots & \ldots &\vdots &\vdots &\ldots &\vdots \\
	1 & X_{I_{2k-1}} & \ldots & X_{I_{2k-1}}^{k-1}  & X_{J_{2k-1}} &\ldots & X_{J_{2k-1}}^{k-1}\\
	\end{pmatrix} .
	\end{equation}


	\begin{prop} \label{prop:cond-for-RS}
		Consider the $[n,k]_q$ RS-code defined by an evaluation vector   $\alpha=(\alpha_1,\ldots,\alpha_n)$.  
			If for every two increasing vectors $I,J\in [n]^{2k-1}$ that agree on at most $k-1$ coordinates, it holds that $\det(V_{I,J}(\alpha)) \neq 0$, then the code can correct any $n-2k+1$ insdel errors.
		Moreover, if the code can correct  any $n-2k+1$ insdel errors, then the only possible vectors in $\text{Kernel}\left(V_{I,J}(\alpha)\right)$ are of the form $(0,f_1,\ldots,f_{k-1},-f_1,\ldots,-f_{k-1})$.
	\end{prop}
\begin{proof}
Assume that the claim does not hold; therefore, there exist two distinct codewords $c\neq c'$ whose LCS is at least $2k-1$, i.e., 
 $c_I=c'_J$ for two increasing vectors $I,J\in [n]^{2k-1}$.  Assume further that $c$ and $c'$ are the encodings of the degree $k-1$ polynomials $f=\sum_if_ix^i$ and $g=\sum_ig_ix^i$, respectively.
If $I_\ell=J_\ell$ for at least $k$ coordinates, then for every such  $\ell$  
$$f(\alpha_{I_\ell})=c_{I_\ell}=c'_{J_\ell}=g(\alpha_{I_\ell})\;.$$
Hence $f\equiv g$, in contradiction to the fact that $c\neq c'$. Thus, we can assume that $I,J$ agree on at most $k-1$ coordinates. In this case, $V_{I,J}(\alpha)$ 
is singular, since the vector $(f_0-g_0,f_1,\ldots,f_{k-1},-g_1,\ldots,-g_{k-1})^t$ is in its right kernel, which contradicts our assumption. From \Cref{lem:lcs} it follows that the code can correct $n-2k+1$ insdel errors. 

\sloppy To prove the moreover part note that the argument above implies that if the code can correct  any $n-2k+1$ insdel errors and $f\neq g$ then the vector $(f_0-g_0,f_1,\ldots,f_{k-1},-g_1,\ldots,g_{k-1})$ is not in the kernel.
\end{proof}

	   In \cite{safavi2002traitor}  Safavi-Naini and Wang identified (almost) the same condition (see \Cref{rem:SNW} below) and used it  in their construction of traitor tracing schemes.	Interestingly, the later work of \cite{tonien2007construction}, which gave a construction of RS-codes capable of decoding from $\log_{k}(n+1)-1$ insdel errors, did not use this condition. In particular, as far as we know, prior to this work the condition in \Cref{prop:cond-for-RS} was not used in order to show the existence of optimal RS-codes.	
	   
	   	The following remark explains the difference between \Cref{prop:cond-for-RS} and the condition in \cite{safavi2002traitor}.
	
	\begin{remark}\label{rem:SNW}
		    The main difference between the condition presented in  \cite{safavi2002traitor} and ours, is that they considered a $2k\times 2k$ matrix and a generalized RS-code. Given evaluation points $(\alpha_1,\ldots,\alpha_n)$ and a vector with nonzero coordinates $(v_1,\ldots,v_n)\in \Fq^n$, the generalized $[n,k]_q$ RS-code is defined as the set of all  vectors $\left(v_1\cdot f(\alpha_1),\ldots,v_n\cdot f(\alpha_n)\right)$, such that $\deg(f)<k$. The  matrix studied in \cite{safavi2002traitor} is:
	    	\begin{equation} \label{eq:mat-SNW}
	V_{I,J}^{v}(\bX)=\begin{pmatrix} 
	v_{I_1} & v_{I_1}\cdot X_{I_1} & \ldots &v_{I_1}\cdot X_{I_1}^{k-1}  & v_{J_1} & v_{J_1}\cdot X_{J_1} &\ldots &v_{J_1}\cdot X_{J_1}^{k-1} \\ 
	v_{I_2} & v_{I_2}\cdot X_{I_2} & \ldots &v_{I_2}\cdot X_{I_2}^{k-1} & v_{J_2} & v_{J_2}\cdot X_{J_2} &\ldots &v_{J_2}\cdot X_{J_2}^{k-1} \\
	\vdots &\vdots & \ldots& \ldots &\vdots &\vdots &\ldots &\vdots \\
	v_{I_{2k}} &v_{I_{2k}}\cdot X_{I_{2k}} & \ldots &v_{I_{2k}}\cdot X_{I_{2k}}^{k-1} &v_{J_{2k}} &v_{J_{2k}}\cdot X_{J_{2k}} &\ldots &v_{J_{2k}}\cdot X_{J_{2k}}^{k-1} 
	\end{pmatrix} .
	\end{equation}
	In our matrix, we saved a coordinate (which leads to optimal codes) as we did not have two columns for the free terms of $f$ and $g$ (as defined in the proof). In contrast, the matrix \eqref{eq:mat-SNW} has a column for the free term of $f$ (the first) and a column for the free term of $g$ (the column $(v_{J_1},\ldots,v_{J_{2k}})$). This also leads to the requirement that $I$ and $J$ are of length $2k$ (they can still agree on at most  $k-1$ indices).
	\end{remark}

    
    \subsection{Optimal Reed-Solomon codes exist}
    \label{sec:randomized-const}
    In this section, we show that  over large enough fields, there exist   RS-codes that attain the half-Singleton bound. Specifically, we show that there exist RS-codes that can decode from a $\delta$ fraction of  insdel errors and have   rate  $\mathcal{R} = (1- \delta)/2 + o(1)$. For convenience, we repeat the statement of \Cref{thm:rs-sz}.	\rssz* 
	
	
	For a vector $I$ and an element $a$, we write $a\in I$ if $a$ appears in one of the coordinates of $I$; otherwise, we write $a\notin I.$ 
	\begin{lemma}
	\label{best-lemma}
	Let $s\geq 2$ be an integer and $I,J\in [n]^s$ two increasing vectors that do not agree on \emph{any} coordinate, i.e., $I_i\neq J_i$ for all $1\leq i\leq s$. Then, there are two distinct indices  $i\neq j\in [s]$ such that $I_i\notin J$ and $J_j\notin I$. 
	\end{lemma}
	\begin{proof}
	W.l.o.g. assume that $I_1<J_1$. Since $J$ is an increasing vector, $I_1\notin J$. In addition, some coordinate among $\{J_1,\ldots,J_s\}$ does not appear in $\{I_2,\ldots,I_s\}$, and any such coordinate is clearly different from $I_1$.
	\end{proof}

	
	\begin{prop} \label{prop:formal-det}
		Let $I,J\in [n]^{2k-1}$ be two  increasing vectors  that agree on at most $k-1$ coordinates.  
		Then, in the expansion of $\det(V_{I,J}(\bX))$ as a sum over permutations, there is a  monomial that is obtained at exactly one of the $(2k-1)!$ different permutations. In particular, its  coefficient is  $\pm 1$, depending on the sign of its corresponding permutation. Consequently, $\det(V_{I,J}(\bX))\neq 0$.
	\end{prop}
	\begin{proof}
		The result  will follow by applying  induction on $k$. For $k=1$, $V_{I,J}(\bX)=1$ and the result follows. For the induction step, assume it holds for $k-1$, and we prove it for $k\geq 2$. 
		Consider  two coordinates $i,j$, determined as follows. If $I$ and $J$ agree on some coordinate, say $j$, then we set $i$ to be such that $I_i\notin J$. 
		If they do not agree on any coordinate, then we  let $i,j$ be the two coordinates guaranteed by Lemma \ref{best-lemma}.

		Next, in the determinant expansion of $V_{I,J}$ as a sum of $(2k-1)!$ monomials, collect all the monomials that are divisible by $X_{I_i}^{k-1}X_{J_j}^{k-1}$, and write them  together as
		$$X_{I_i}^{k-1}X_{J_j}^{k-1}f(\bX),$$ 
		for some polynomial $f$ in the variables $(X_\ell : \ell \in (I\setminus\{I_i\}) \cup (J\setminus\{J_j\}))$. Note that the choice of $i$ and $j$ guarantees that such monomials exist.  Observe that any monomial in  the determinant expansion of $V_{I,J}$ that is divisible by $X_{I_i}^{k-1}X_{J_j}^{k-1}$ must be obtained by picking the $(i,k)$ and the $(j,2k-1)$ entries in the matrix  \eqref{eq:mat-lcs-eq}. Hence,
		  $f$ equals  the determinant of the submatrix $V'_{I,J}$ obtained by removing   rows $i, j$  and columns $k,2k-1$ from $V_{I,J}$. Note that $V'_{I,J}$ is a matrix satisfying the conditions of the claim: it is  a  $(2k-3)\times(2k-3)$ matrix defined by two increasing vectors of length $2k-3$  that agree on at most $k-2$ coordinates. Indeed, $i$ and $j$ were chosen so that by removing them we remove one agreement, if such existed.
		  
		  Hence, by the induction hypothesis $\det(V'_{I,J})$ has a monomial $m$ (with a $\pm 1$ coefficient) that is uniquely obtained among the $(2k-3)!$ different monomials.
		  Therefore,  $X_{I_i}^{k-1}X_{J_j}^{k-1}m$ is a monomial of $X_{I_i}^{k-1}X_{J_j}^{k-1}f$ with a $\pm 1$ coefficient. Since there is no other way to obtain  this monomial in the determinant expansion of $V_{I,J}$, this monomial is uniquely obtained   in $\det(V_{I,J})$, and  the result follows. 
	\end{proof}
	We proceed to  prove \Cref{thm:rs-sz} by a standard application of the  Schwartz-Zippel lemma.

	\begin{proof}[Proof of \Cref{thm:rs-sz}]
		Define
		$$F(\bX)=\prod_{i< j}(X_i-X_j)\prod_{I,J}\det(V_{I,J}(\bX)),$$
where the second product runs over all possible pairs of increasing vectors that agree on no more than $k-1$ coordinates. Clearly, by \Cref{prop:formal-det},  $F(\bX)$ is  a nonzero polynomial in the ring 
$ \mathbb{Z}[\bX]$.  Next, we make two observations regarding the polynomial $F$. First, 
since there are $\binom{n}{2k-1}$ increasing vectors, and the degree of each   $\det(V_{I,J}(\bX))$ is at most $k(k-1)$, it follows that
\[
		\deg(F) \leq n^2 +\binom{n}{2k - 1} ^2 \cdot k(k-1) < n^{4k-2} \;.
		\]
	
    Second, as each $\det(V_{I,J}(\bX))$ is a nonzero polynomial with nonzero coefficients bounded in absolute values by $(2k-1)!$,  the absolute value of any nonzero coefficients of $F$ is at most 
    \[
    ((2k-1)!)^{\binom{n}{2k-1}^2}\leq ((2k-1)!)^{\frac{n^{4k-2}}{((2k-1)!)^2}} < e^{n^{4k-2}}. 
    \]
    We claim that there is a prime $q$ in the range $[n^{4k-2},2n^{4k-2}]$ that does not divide at least one of the nonzero coefficients of the polynomial $F$. Indeed, consider a nonzero coefficient of $F$, and assume towards a contradiction that it is divisible by all such primes. Then, by the growth rate of the primorial function, the absolute value of the coefficient is  $\Omega(e^{n^{4k-2}(1+o(1))})$, in contradiction.	
	Now, it is  easy to verify that   $F$ is also a nonzero  polynomial  in $\Fq[\bX]$,  since the monomial whose nonzero coefficient is not divisible by $q$ does not vanish. Therefore, by the Schwarz-Zippel-Demillo-Lipton lemma, there is an assignment $\alpha=(\alpha_1,\ldots,\alpha_n)$ to $\bX$ for which $F(\alpha)\neq 0 \mod q$. 
	This assignment clearly corresponds to  $n$ \emph{distinct} evaluation points, which by \Cref{prop:cond-for-RS}, define  an $[n,k]_q$ RS-code
	that can correct any $n-2k+1$ worst-case insdel errors, as claimed.
		%
	\end{proof}
	
    We remark again  that \Cref{thm:rs-sz}  merely shows the existence  of $[n,k]_q$ RS-codes that can decode from the maximum number of  insdel errors over a field of size  $q=O\left(n^{4k-2}\right)$. 
	Further, the above argument is a standard union bound over all variable assignments that make the matrix defined in \eqref{eq:mat-lcs-eq} to be singular. This by no means implies that such a large finite field is necessary. For example,  a similar union-bound argument that shows the existence of MDS codes would require an exponentially large field for codes with a constant rate. In contrast, it is well-known that MDS codes over linear field size exist (e.g., RS-codes). It would be interesting to explicitly construct codes with the same or even better parameters than the ones given in  \Cref{thm:rs-sz}.   
	Unfortunately,  we could not construct such codes, and this is left as an open question for further research.
	Nonetheless, in the next section, we provide a deterministic construction of an RS-code for any admissible $n,k$, at the expense of a larger field size than the one guaranteed by \Cref{thm:rs-sz}.

	
	\section{Deterministic construction for any $k$}
	\label{sec:det-any-k}
		In this section, we give   our main construction of an $[n,k]$ RS-code that can correct any $n-2k+1$  insdel errors. Specifically, we prove  Theorem \ref{thm:rs-explicit-const} which is restated for convenience
	\rsExpConst*
	\begin{remark}
	    The downside of this construction is the field size  $q=n^{k^{O(k)}}$, which  renders it to run in polynomial time only  for  $k = O(\log(n)/\log(\log(n)))$. For larger values of $k$, the representation of each field element requires a super polynomial number of bits.
	\end{remark}
The Mason–Stothers theorem \cite{mason,stothers1981polynomial} is a result about  polynomials that satisfy a non-trivial linear dependence, which  is analogous to the well-known \emph{abc conjecture} in number theory \cite{Masser,oesterle1988nouvelles}. Our main tool is one of the many extensions  in the literature to the Mason–Stothers theorem. For stating the theorem we need the following notation: For a polynomial $Y(x)\in \mathbb{F}[x]$ over a field with $\charac(\mathbb{F})=p\neq 0$,  denote by  $\nu (Y(x))$  the number of distinct roots of $Y(x)$ with multiplicity  not divisible by $p$. 

	\begin{thm}[``Moreover part'' of Proposition 5.2 in \cite{vaserstein2003vanishing}] \label{thm:abc-poly}
		Let $m \geq 2$ and $Y_0(x) = Y_1(x) + \ldots + Y_m(x)$ with $Y_j(x)\in \Fp [x]$. Suppose that $\gcd (Y_0(x), \ldots, Y_m(x)) = 1$, and  that $Y_1(x), \ldots, Y_m(x)$ are linearly independent over $\Fp (x^p)$.\footnote{$\Fp (x^p)$ is the field of rational functions in $x^p$. Namely, its elements are $f(x^p)/g(x^p)$ where $f(x),g(x) \in \Fp[x]$ and $g(x)\not \equiv 0$.} Then, 
		\[
		\deg (Y_0(x)) \leq -\binom{m}{2} + (m-1)\sum_{j=0}^{m} \nu (Y_j(x))\;.
		\]
	\end{thm}

	\begin{cnst} \label{cnst:abc-rs}
		Let $k$ be a positive integer and set $\ell = ((2k)!)^2$.
		Fix a finite field $\Fp$ for a prime $p > k^2 \cdot \ell$ and let $n$ be an integer such that $2k-1< n \leq p$. Let $\Fq$ be a  field extension of $\Fp$ of  degree $k^2\cdot \ell$ and let $\gamma\in \Fq$ be such that     $\Fq = \Fp(\gamma)$. Hence, each element of $\Fq$ can be represented as a polynomial in $\gamma$, of degree less than $k^2\ell$, over $\Fp$. Define the $[n,k]_q$ RS-code by setting 
		 $\alpha_i := (\gamma-i)^{\ell}$ for  $1\leq i\leq n$.
	\end{cnst}
	%
	
	\begin{prop} \label{prop:abc-rs}
		The $[n,k]_q$ RS-code defined in \Cref{cnst:abc-rs} can correct  any $n - 2k +1$ worst case insdel errors.
	\end{prop}
	\begin{proof}
	Let  $I,J\in [n]^{2k-1}$ be two increasing vectors  that agree on at most $k-1$ coordinates. By Proposition \ref{prop:cond-for-RS} it is enough to show that $V_{I,J}(\alpha)$ is non-singular, for every such $I,J$.  By the Leibniz formula, $\det(V_{I,J}(\alpha))$ is a sum  of $(2k-1)!$ terms corresponding to the different permutations. Denote these terms as $P_i(\gamma)$ for $i=0,\ldots,(2k-1)!-1$. Each of the terms is a product of the sign of the corresponding permutation with some $2k-1$ elements of the form   $(\gamma - s)^{\ell \cdot j}$, for some $s\in I\cup J$ and $0\leq j \leq k-1$. 
	Assume towards a contradiction that $\det(V_{I,J}(\alpha)) = 0$ in $\Fq$, i.e., 
	\begin{equation}
	\label{good-eq2}
		\det(V_{I,J}(\alpha)) = P_0(\gamma) + \ldots + P_{(2k-1)!-1}(\gamma) = 0 \;,
	\end{equation}
	in $\Fq$. By viewing every term in \eqref{good-eq2} as a univariate polynomial in $\gamma$ over $\Fp$, one can verify that, for any $ j$,  $\deg (P_j)$ = $\ell \cdot k(k-1)<k^2\ell$. As $\Fq$ is an extension of $\Fp$ of degree $k^2\ell$, it follows that \eqref{good-eq2} holds also in $\Fp[\gamma]$, the ring of polynomials in the variable $\gamma$ over $\Fp$.
	By Proposition \ref{prop:formal-det} the determinant of the variable matrix \eqref{eq:mat-lcs-eq}  has a  monomial  that is uniquely obtained and therefore has  a $\pm 1$  coefficient. Assume, without loss of generality, that 
	$P_0$ is the  image of this monomial under   the mapping defined by the assignment $X_i\mapsto (\gamma-i)^\ell$. Note  that since this mapping is  injective on the  set of monomials, no other monomial is mapped to  a scalar multiple of  $P_0$. 
	In other words,  $P_0$ and $P_i$ are linearly independent for any $i\geq 1$. Assume further that (without loss of generality) $P_1,\ldots, P_m$ is a minimal subset among $\{P_i\}_{i\geq 1}$ that spans $P_0$ over $\Fp$. The existence of such a set follows from \eqref{good-eq2}. 	Hence, we can write    
\begin{equation}
\label{last-eq}
    P_0=\sum_{i=1}^ma_iP_i,\text{ where } a_i\in \mathbb{F}_p \backslash \{0\}.
\end{equation}
	Clearly, by  minimality, $P_1,\ldots,P_m$ are linearly independent over $\Fp$. Further, $m\geq 2$, since otherwise there would be an $i>0$ such that $P_i$  is a multiple of $P_0$.
	
	Since  the $P_i$'s are of degree $\ell k(k-1)$, and $P_0$ was obtained from a  unique monomial in the determinant expansion, it follows that  the greatest common divisor  
	$Q:=\gcd(P_0,\ldots,P_m)$
	has degree at most $\ell(k(k-1)-1)$. By dividing \eqref{last-eq} by $Q$ we have 
	\begin{equation}
    \label{last-eq2}
	    \overline{P_0}=\sum_{i=1}^ma_i\overline{P_i},
	\end{equation}
	where $\overline{P_i}=P_i/Q$. 
	We will need the following claim, whose proof is deferred to the end of this section.
		\begin{claim} \label{clm:poly-ind}
			The polynomials $\overline{P_1}, \ldots, \overline{P_m}$ are linearly independent over $\Fp(\gamma^p)$.
		\end{claim}
	
	The contradiction will follow by invoking  Theorem \ref{thm:abc-poly}. Towards this end note that (i) By Claim \ref{clm:poly-ind} the polynomials  $\overline{P_1},\ldots,\overline{P_m}$ are linearly independent of $\mathbb{F}_p(\gamma^p)$ (ii) $\gcd(\overline{P_0},\ldots,\overline{P_m})=1$, and  (iii)     $\nu (\overline{P_j}) \leq  2k-2$, as $P_j$ is the multiplication of $2k-2$  non-constant polynomials, each having  a single root.
		Thus, by  \eqref{last-eq2} and \Cref{thm:abc-poly}
		\begin{align*}
		\ell\leq \deg(P_0)-\deg(Q)= \deg(\overline{P_0}) 
		&\leq -\binom{m}{2} + \left(m - 1 \right) \cdot \sum_{i=1}^m \nu(\overline{P_j})\\
		&< m^2(2k-2)\\
		&\leq ((2k-1)!)^2\cdot (2k-2)\;,
		\end{align*}
		which is a contradiction by the choice of $\ell$. This completes the proof. 
	\end{proof}
	It remains to prove Claim \ref{clm:poly-ind}.

	\begin{proof}[Proof of \Cref{clm:poly-ind}]
		 Assume towards a contradiction that there exist $\lambda_1,\ldots, \lambda_m \in \Fp(\gamma^p)$ not all zero,  such that 
		\begin{equation} \label{eq:abc-ind-poly}
		\sum_{j=1}^{m} \lambda_j(\gamma^p) \overline{P_j}(\gamma) = 0\;.
		\end{equation}
		By clearing the denominators of the $\lambda_j$'s and any  common factor they might have, we can assume that the $\lambda_j$'s are  polynomials in the variable $\gamma^p$ with no common factors. Since $\deg(\overline{P_j})\leq \deg(P_j)<p$, we get by reducing  \eqref{eq:abc-ind-poly} modulo $\gamma^p$ that 
		$$\sum_{j=1}^{m} \lambda_j(0) \overline{P_j}(\gamma) = 0\;.$$
		Note that  $\lambda_j(0)\neq 0$ for some $j$, since otherwise $\gamma^p$ would be a common factor of the $\lambda_i$'s. Hence,   $\overline{P_1},\ldots,\overline{P_m}$ are linearly dependent over $\Fp$, which contradicts the fact that  $P_1,\ldots,P_m$  are linearly independent over $\mathbb{F}_p$.  
	\end{proof}
	By setting  $n = p$ in \Cref{cnst:abc-rs} it follows that  the field size of  Construction \ref{cnst:abc-rs} is roughly $n^{k^{O(k)}}$ which is much worse than the field size guaranteed by the existential result in Theorem \ref{thm:rs-sz}. Note, however, that the construction runs in polynomial time  for  RS-codes with dimension $O(\log(n)/\log(\log(n)))$. The proof of \Cref{thm:rs-explicit-const} immediately follows.

	\section{Explicit construction for $k=2$ with quartic field size}
    \label{sec:det-k-2}

	In this section we prove  \Cref{thm:rs-twodim-const}, which is restated for convenience.
	\rsTwoDimConst*
	The proof of \Cref{thm:rs-twodim-const} requires  the notion of  
	Sidon spaces, which     were  introduced in a work of Bachoc, Serra and Z{\'e}mor
	\cite{bachoc2017analogue} in the study  of an analogue of Vosper's theorem for finite fields. 
	Later, Roth, Raviv and Tamo  gave an explicit construction of  Sidon spaces and used it to provide a construction of  cyclic subspace codes \cite{roth2017construction}. Our construction  relies on the construction of Sidon spaces of \cite{roth2017construction}, which was also recently used in \cite{DBLP:conf/pkc/RavivLT21} to construct a public-key cryptosystem. We believe that  Sidon spaces in general, and  specifically the construction of \cite{roth2017construction}, might find  more applications in coding theory and cryptography in the future.  We begin with a formal definition of a Sidon space. 
	\begin{defi}
		An $\Fq$ linear  subspace $S\subseteq \mathbb{F}_{q^n}$  is called a \emph{Sidon space} if for any nonzero elements $a,b,c,d \in S$ such that   $ab=cd$, it holds that that $$\{a \Fq, b \Fq\} = \{c \Fq, d \Fq\},$$ where $x \Fq=\{x  \cdot \alpha: \alpha\in \Fq\}$ . 
	\end{defi}
	A  Sidon space $S$ has the following interesting property, from which it draws its name: Given the product  $a\cdot b$ of   two nonzero elements $a,b\in S $, one can uniquely factor it to its two factors $a,b$ from $S$,  up to a multiplication by a scalar from the base field. Clearly, this is the best one can hope for, since for any nonzero $\alpha\in \Fq$ the product of the  elements $\alpha\cdot a, b/\alpha\in S$ also equals $a\cdot b$. A Sidon space can be viewed as a multiplicative analogue to the well-known notion of \emph{Sidon sets}, which is a common object of study in combinatorics, see e.g.   \cite{erdos1941problem}. 
	
	We proceed to present the  construction of a Sidon space  given in  \cite{roth2017construction}.
	\begin{thm}[Construction 15, Theorem 16 in \cite{roth2017construction}] \label{thm:Sidon-tamo}  \label{thm:sidon-space}
		Let $q\geq 3$ be a prime power, $m\in \mathbb{N}$, and   $n = 2m$. Then, there exists an explicit $\gamma \in \mathbb{F}_{q^n}$ such that  $S = \{u + u^q \cdot \gamma \mid u\in \mathbb{F}_{q^m} \}$ is  an $m$-dimensional Sidon space over $\Fq$.
	\end{thm}
	Another component in our construction  is the ``long'' ternary code with   minimum distance of at least $5$, given in \cite{gashkov1986linear}. We note that we could also use the codes given in \cite{danev2008family}.
	
	\begin{thm} \cite{gashkov1986linear} \label{thm:bch}
		For every $m \geq 1$, there exits an explicit $[(3^m+1)/2, (3^m+1)/2 - 2m]_3$ linear code with minimum distance  at least $5$.
	\end{thm}
	We next combine the above two algebraic objects and construct an RS-code with the desired properties. 
	\begin{cnst}
	\label{RS k=2}
	For $q = 3$ and  $m\in \mathbb{N}$. Let 
	$S \subset \mathbb{F}_{3^{4m}}$ be a  $2m$-dimensional Sidon space over $\mathbb{F}_3$ as guaranteed by \Cref{thm:sidon-space}. Let $s_1, \ldots, s_{2m}$ be a basis of $S$. 
	Let $H = (h_{i,j})$ be a  ${(2m) \times ((3^m +1)/2)}$ parity check matrix of the code given  in \Cref{thm:bch}.
	Our  $[n,2]_{3^{4m}}$ RS-code of length  $n=(3^m +1)/2$ is defined by the evaluation points 
	\[
	\alpha_j = \sum_{i=1}^{2m}s_i h_{i,j} \text{ for }  1\leq j \leq (3^m+1)/2  \;.
	\]
	In other words, we can think of our evaluation points as the $n$ coordinates of the vector $\alpha=(s_1,\ldots,s_{2m})\cdot H$.
	\end{cnst}
	
The following property of the evaluation points $\alpha_j$ follows easily from \Cref{thm:bch}. 	
	\begin{lemma}
		\label{158}
		Any four distinct $\alpha_j$'s  are linearly independent over $\mathbb{F}_3$.
	\end{lemma}
	
	\begin{proof}
		Consider four distinct $\alpha_j$'s, say  $\alpha_1, \alpha_2, \alpha_3, \alpha_4$, and  assume towards a contradiction that there exist $\beta_1,\ldots,\beta_4\in \mathbb{F}_3$ not all zero,   such that 
		$\sum_{i=1}^4\beta_i\alpha_i=0$. Then 
$$
		0=\sum_{i=1}^4\beta_i\alpha_i
		=\sum_{i=1}^4\beta_i\sum_{j=1}^{2m}s_j h_{j,i}
		=\sum_{j=1}^{2m}s_j\sum_{i=1}^{4}\beta_i h_{j,i} \;.
		$$
		Since the $s_j$'s are linearly independent over $\mathbb{F}_3$ it follows that $\sum_i\beta_ih_{j,i}=0$ for every $j=1,\ldots,2m$. Hence, the four columns  $h_1,h_2, h_3, h_4$ of $H$ are  linearly dependent  over $\mathbb{F}_3$, which contradicts the fact that the minimum distance of the code checked by $H$ is at least $5$.
	\end{proof}
We proceed to prove that the constructed RS-code can decode from  the maximum number of  insdel errors. 
	\begin{thm} \label{prop:k-2-explicit}
		The $\left[n,2\right]_{3^{4m}}$ RS-code given in Construction \ref{RS k=2}  can correct any $n-3$ worst case  insdel errors. 
	\end{thm}
	\begin{proof}
				Assume towards a contradiction  that this is not the case. \Cref{prop:cond-for-RS} implies that there must exist two triples of distinct  evaluation points $(x_1,x_2,x_3),(y_1,y_2,y_3)$, that agree on at most one coordinate, such that 
		\[
		\left|
		\begin{pmatrix}
		1 &x_1 &y_1 \\ 
		1 &x_2 &y_2 \\
		1 &x_3 &y_3 \\
		\end{pmatrix} 
		\right|
		= 0 
		\;.
		\]
		Equivalently, 
		 $(y_1 - y_2)(x_2 - x_3) = (y_2 - y_3)(x_1 - x_2)$. Since the $x_i$'s  are distinct elements  of the Sidon space $S$, $x_2 - x_3$ and  $x_1 - x_2$ are \emph{nonzero} elements in $S$. Similarly,  $y_1 - y_2$ and $y_2 - y_3$ are nonzero elements in  $S$. By definition of Sidon spaces,  there exists a nonzero $\lambda\in \mathbb{F}_3$ such that 
		$$\lambda(y_1-y_2)=y_2-y_3\text{ or } \lambda(y_1-y_2)=x_1-x_2,$$
		which  contradicts Lemma \ref{158}. Indeed, each of the  equations implies  a nontrivial linear dependence over $\mathbb{F}_3$ between at least three and at most four evaluation points (here we used the facts that the elements of each triple are distinct and that the two triples  agree on at most one coordinate).  
		%
		%
	\end{proof}
	We conclude this section with the proof of  
	\Cref{thm:rs-twodim-const}. 
	
	\begin{proof}[Proof of \Cref{thm:rs-twodim-const}]
		By \Cref{prop:k-2-explicit}, the code given in Construction \ref{RS k=2} is an RS-code of length $n = (3^m+1)/2$,  defined over  the field $\mathbb{F}_{3^{4m}}$,  which is of order $O(n^4)$, as claimed.
	\end{proof}

	\phantom{We remark again that this proof is an existence proof, namely, we show that there are $[n,k]_q$ RS-codes that can correct from $n-2k+1$ worst-case insdel errors where $q = n^{O(k)}$. Unfortunately, for the case of $k > 2$, we did not manage to provide explicit construction that achieves this field size. }

\subsection{A lower bound on the field size}
	\label{sec:rs-lower-bound}
	In Section \ref{sec:randomized-const} we proved the existence  of   optimal $[n,k]_q$ RS-codes for worst-case insdel errors  over fields of  size  $q=n^{O(k)}$.
	This section complements this result by providing a lower bound on the field size for such codes. Specifically, we ask how large must $q$ be in an $[n,k]_q$ RS-code that can correct from $n-2k+1$ worst-case insdel errors. We prove the following. 
	\RSLowerBound*
	
	\begin{proof}
		Consider an $[n,k]_q$ RS-code, defined by evaluation points  $\alpha_1, \ldots, \alpha_n$,  that can correct any $n-2k+1$ insdel errors.
		For a \emph{non-constant} polynomial $f$ of degree less than $k$ let $\mathcal{V}_f$ be the set of all subsequences of the codeword corresponding to $f$, of length $2k-1$:
		$$\mathcal{V}_f=\{(f(\alpha_{i_1}),\ldots, f(\alpha_{i_{2k-1}})):1\leq i_1<\ldots<i_{2k-1}\leq n\}\subseteq \Fq^{2k-1}.$$
		By Lemma \ref{lem:lcs}, since the code can decode from any $n-2k+1$ insdel errors, the sets $\mathcal{V}_f$ and  $\mathcal{V}_g$ for  two distinct polynomials $f, g$, are disjoint. Therefore,\footnote{This equation remains true also if we include the constant polynomials.} 
		\begin{equation}
		\label{good-eq4}
		\sum_{1\leq \deg(f)<k} |\mathcal{V}_f|\leq q^{2k-1}.
		\end{equation}
		Next, we provide a lower bound on the size of  $\mathcal{V}_f$.
		For any non-constant polynomial $f$, of degree less than $k$,  and any $a\in \Fq$ there are at most $k-1$ indices $i$ such that $f(\alpha_i)=a$. Thus, for a fixed vector $(a_1,\ldots, a_{2k-1}) \in \mathcal{V}_f$ there are at most $(k-1)^{2k-1}$ increasing vectors of indices $(i_1,\ldots,i_{2k-1})$ such that 
		$$(f(\alpha_{i_1}),\ldots, f(\alpha_{i_{2k-1}}))=(a_1,\ldots ,a_{2k-1}).$$ 
		Therefore $|\mathcal{V}_f|\geq \binom{n}{2k-1}(k-1)^{-(2k-1)}$.  Combined with \eqref{good-eq4} we have 
		\[
		\left( \frac{1}{k-1} \right)^{2k-1} \cdot \binom{n}{2k-1} \cdot \left(q^k - q\right) \leq q^{2k-1} \;,
		\]
		By rearranging and the fact that  
		 $q^{2k-1}/{(q^k - q)} \leq 2q^{k-1}$ for $q,k\geq 2$, we have  
		\[
	 \left(\frac{1}{2}\right)^{\frac{1}{k-1}} \left(\frac{n}{(2k-1)(k-1)}\right)^{\frac{2k-1}{k-1}}\leq q \;.\qedhere
		\] 
	\end{proof}
	As one can easily verify, this bound is rather weak, as it provides an improvement over the trivial lower bound of $q\geq n$ only for the vanishing  rate regime of  $k=O(n^{1/4})$. 
	For  codes of dimension $2$, the bound implies  
	 $q = \Omega(n^3)$, and it slowly degrades as one increases $k$. Nevertheless, it is always  at least $\Omega(n^2)$ for any constant $k$.
	 It is interesting to note that by combining 
	Proposition \ref{prop:lower-bound}  and Theorem \ref{prop:k-2-explicit}
	 it follows  that an $[n,2]_q$ RS-code that can decode from $n-3$ insdel errors requires that $\Omega(n^3)\leq q \leq O(n^4)$. Determining the  minimum possible value  of $q$ for this case is an interesting open problem. 
\section{Open questions}

\label{sec:open-que}
    This paper studies the performence of RS codes against insdel errors. We showed that there are RS-codes are optimal against insdel errors, i.e., they achieve the half-Singleton bound. We also construct explicit RS codes that achieve this bound and as far as we know, this is the first linear code that is shown to achieve this bound.
    
    As discussed, \Cref{cnst:abc-rs} is not optimal in terms of the field size. It is a fascinating open question to find an RS-code with an optimal field size. Specifically, the challenge is to construct an RS-code that can correct from any $n-2k+1$ insdel errors, over a field of size $O(n^{O(k)})$ (\Cref{thm:rs-sz} proves the existence of such codes).
	
	The lower bound on the field size proved in \Cref{prop:lower-bound} is far from giving a full picture of the tradeoff between dimension and  field size. The natural open question is to significantly improve our lower bound or provide a better upper bound.
	
	Finally, another interesting question is to provide an efficient decoding algorithm for our constructions of RS-codes.  
\section*{Acknowledgement}
We thank Shu Liu, Ivan Tjuawinata and Chaoping Xing for spotting an inaccuracy in the original statement of \autoref{prop:cond-for-RS}.

	\bibliographystyle{alpha}
	\bibliography{RSinsdel}

\end{document}